\documentclass[11pt]{article}
\pdfoutput=1
\usepackage{latexsym}
\usepackage{amsmath}
\usepackage{amssymb}
\usepackage{amsfonts,amsthm}
\usepackage[top=1in, bottom=1in, left=1in, right=1in]{geometry}
\usepackage{tikz}
\usetikzlibrary{arrows,decorations.pathmorphing,decorations.shapes,backgrounds,fit}
\pgfkeys{tikz/cc1/.style={dotted} }
\pgfkeys{tikz/cc4/.style={dashed} }
\pgfkeys{tikz/cc3/.style={decorate,decoration=bumps} }
\pgfkeys{tikz/cc2/.style={decorate,decoration=snake} }
\pgfkeys{tikz/solid/.style={line width=2pt} }
\pgfkeys{tikz/solid2/.style={line width=1pt} }
\pgfkeys{tikz/nodostile1/.style={draw,line width=1pt}}
\pgfkeys{tikz/nodostile2/.style={draw,line width=2pt}}
\pgfkeys{tikz/arcostile1/.style={very thick} }
\pgfkeys{tikz/arcostile2/.style={red,very thick} }
\pgfkeys{tikz/arcostile3/.style={thick} }
\pgfkeys{tikz/arcostile4/.style={line width=2pt} }

\newcommand{\ignore}[1]{}

\newtheorem{theorem}{Theorem}[section]
\newtheorem{lemma}[theorem]{Lemma}

\newtheorem{corollary}[theorem]{Corollary}

\usepackage{array}
\newcolumntype{C}[1]{>{\centering\let\newline\\\arraybackslash\hspace{0pt}}m{#1}}

\begin{document}
\title{\bf 2-Edge Connectivity in Directed Graphs}
\author{Loukas Georgiadis$^{1}$ \and Giuseppe F. Italiano$^{2}$ \and Luigi Laura$^{3}$ \and Nikos Parotsidis$^{1}$}
\maketitle
\thispagestyle{empty}

\begin{abstract}
Edge and vertex connectivity are fundamental concepts in graph theory. While they have been thoroughly studied in the case of undirected graphs, surprisingly not much has been investigated for directed graphs.
In this paper we study $2$-edge connectivity problems in directed graphs and, in particular, we consider the computation of the following natural relation: We say that two vertices $v$ and $w$ are $2$-\emph{edge-connected} if there are two edge-disjoint  paths from $v$ to $w$ and two edge-disjoint  paths from $w$ to $v$.
This relation partitions the vertices into blocks such that all vertices in the same block are $2$-edge-connected. Differently from the undirected case, those blocks do not correspond to the $2$-edge-connected components of the graph. We show how to compute this relation in linear time so that we can report in constant time if two vertices are $2$-edge-connected.
We also show how to compute in linear time a sparse certificate for this relation, i.e., a subgraph of the input graph that has $O(n)$ edges and maintains the same $2$-edge-connected blocks as the input graph, where $n$ is the number of vertices.
\end{abstract}

\footnotetext[1]{Department of Computer Science \& Engineering, University of Ioannina, Greece. E-mail: \texttt{\{loukas,nparotsi\}@cs.uoi.gr}.}
\footnotetext[2]{Dipartimento di Ingegneria Civile e Ingegneria Informatica, Universit\`a di Roma ``Tor Vergata'', Roma, Italy. E-mail: \texttt{giuseppe.italiano@uniroma2.it}.
Partially supported by MIUR, the Italian Ministry of Education, University and Research, under Project AMANDA
(Algorithmics for MAssive and Networked DAta).}
\footnotetext[3]{Dipartimento di Ingegneria Informatica, Automatica e Gestionale, ``Sapienza'' Universit\`a di Roma, Roma, Italy. E-mail: \texttt{laura@dis.uniroma1.it}.}

\section{Introduction}
\label{sec:introduction}


Let $G=(V,E)$ be an \emph{undirected} (resp., \emph{directed}) graph, with $m$ edges and $n$ vertices. Throughout the paper, we use interchangeably the term directed graph and digraph.
Edge and vertex connectivity are fundamental concepts in graph theory with numerous practical applications~\cite{digraphs:jensen-gutin,connectivity:nagamochi-ibaraki}. As an example, we mention the computation of disjoint paths in routing and reliable communication, both in  undirected and  directed graphs \cite{DAC:DAC612,ind_st:ir}. 

We assume that the reader is familiar with the standard graph terminology, as contained for instance in~\cite{algorithms:clr}.
An \emph{undirected path} (resp., \emph{directed path}) in $G$ is a  sequence of vertices $v_1$, $v_2$, $\ldots$, $v_k$, such that edge $(v_i,v_{i+1})\in E$ for $i = 1, 2, \ldots , k-1$.
An undirected graph $G$ is \emph{connected}  if there is an undirected path from each vertex to every other vertex.
The \emph{connected components} of an undirected graph are its maximal connected subgraphs.
A directed  graph $G$ is \emph{strongly connected} if there is a directed path from each vertex to every other vertex.
The \emph{strongly connected components} of a directed graph are its maximal connected subgraphs.

Given an undirected graph $G=(V,E)$,
an edge is a \emph{bridge}  if its removal increases the number of connected components of $G$. Graph $G$ is $2$-edge-connected  if it has no  bridges. The $2$-edge-connected components of $G$ are its maximal $2$-edge-connected subgraphs.
Two vertices $v$ and $w$ are $2$-edge-connected if the removal of any edge leaves them in the same connected component:
we denote this relation by  $v \leftrightarrow_{\mathrm{2e}} w$. Anologous definitions can be given for $2$-vertex connectivity.
In particular, a vertex is an  \emph{articulation point} if its removal increases the number of connected components of $G$. A graph $G$ is $2$-vertex-connected if it has no  articulation points. The $2$-vertex-connected components of $G$ are its maximal $2$-vertex-connected subgraphs.
Note that this allows for degenerate $2$-vertex-connected components consisting of one single edge.
Two vertices $v$ and $w$ are  $2$-vertex-connected if the removal of any vertex different from $v$ and $w$ leaves them in the same connected component:
we denote this relation by  $v \leftrightarrow_{\mathrm{2v}} w$.
By Menger's Theorem~\cite{menger}, $v \leftrightarrow_{\mathrm{2e}} w$
(resp., $v \leftrightarrow_{\mathrm{2v}} w$, with $v$ and $w$ being non-adjacent)
if and only if there are two edge-disjoint  (resp., internally vertex-disjoint) paths between $v$ and $w$. It is easy to show that
$v \leftrightarrow_{\mathrm{2e}} w$ (resp., $v \leftrightarrow_{\mathrm{2v}} w$) if and only if $v$ and $w$ are in a same
$2$-edge-connected (resp., non-degenerate $2$-vertex-connected) component. All bridges, articulation points, $2$-edge- and $2$-vertex-connected components of undirected graphs can be computed in linear time essentially by the same algorithm~\cite{dfs:t}.

The notions of 2-edge and 2-vertex connectivity can be naturally extended to directed graphs (digraphs).
Given a digraph $G$, an edge (resp., a vertex) is a \emph{strong bridge} (resp., a \emph{strong articulation point}) if its removal increases the number of strongly connected components of $G$. A digraph $G$ is $2$-edge-connected (resp., $2$-vertex-connected) if it has no  strong bridges (resp., strong articulation points). The $2$-edge-connected (resp., $2$-vertex-connected) components of $G$ are its maximal $2$-edge-connected (resp., $2$-vertex-connected) subgraphs.
Again, this allows for degenerate $2$-vertex-connected components consisting of two mutually adjacent vertices (i.e., two vertices $v$ and $w$ and the two edges $(v,w)$ and $(w,v)$).
Similarly to the undirected case, we say that two vertices $v$ and $w$ are $2$-edge-connected (resp., $2$-vertex-connected), and we denote this relation by  $v \leftrightarrow_{\mathrm{2e}} w$ (resp. $v \leftrightarrow_{\mathrm{2v}} w$), if the removal of any edge (resp., any vertex different from $v$ and $w$) leaves $v$ and $w$ in the same strongly connected component. It is easy to see that $v \leftrightarrow_{\mathrm{2e}} w$ (resp. $v \leftrightarrow_{\mathrm{2v}} w$, with $v$ and $w$ not being mutually adjacent) if and only if there are two edge-disjoint  (resp., vertex-disjoint) directed paths from $v$ to $w$ and two edge-disjoint  (resp., vertex-disjoint) directed paths from $w$ to $v$. (Note that a path from $v$ to $w$ and a path from $w$ to $v$ need not be edge-disjoint or vertex-disjoint).
We define a \emph{$2$-edge-connected block} (resp., \emph{$2$-vertex-connected block}) of a digraph $G=(V,E)$ as a maximal subset $B \subseteq V$ such that $u \leftrightarrow_{\mathrm{2e}} w$ (resp., $u \leftrightarrow_{\mathrm{2v}} w$) for
all $u, v \in B$.
It can be easily seen that, differently from undirected graphs, in digraphs $2$-edge- and $2$-vertex-connected blocks do not correspond to  $2$-edge-connected and $2$-vertex-connected components, as illustrated in Figure~\ref{fig:graph1}.
Furthermore, these notions seem to have a much richer (and more complicated) structure in digraphs. Just to give an example,  we observe that while in the case of undirected connected graphs the $2$-edge-connected components (which correspond to the $2$-edge-connected blocks) are exactly the connected components left after the removal of all bridges, for directed strongly connected graphs the $2$-edge-connected components, the $2$-edge-connected blocks, and the strongly connected components left after the removal of all  strong bridges are not necessarily the same (see Figure~\ref{fig:example1}).

\begin{figure}[t!]
\begin{center}
\includegraphics[width=\textwidth]{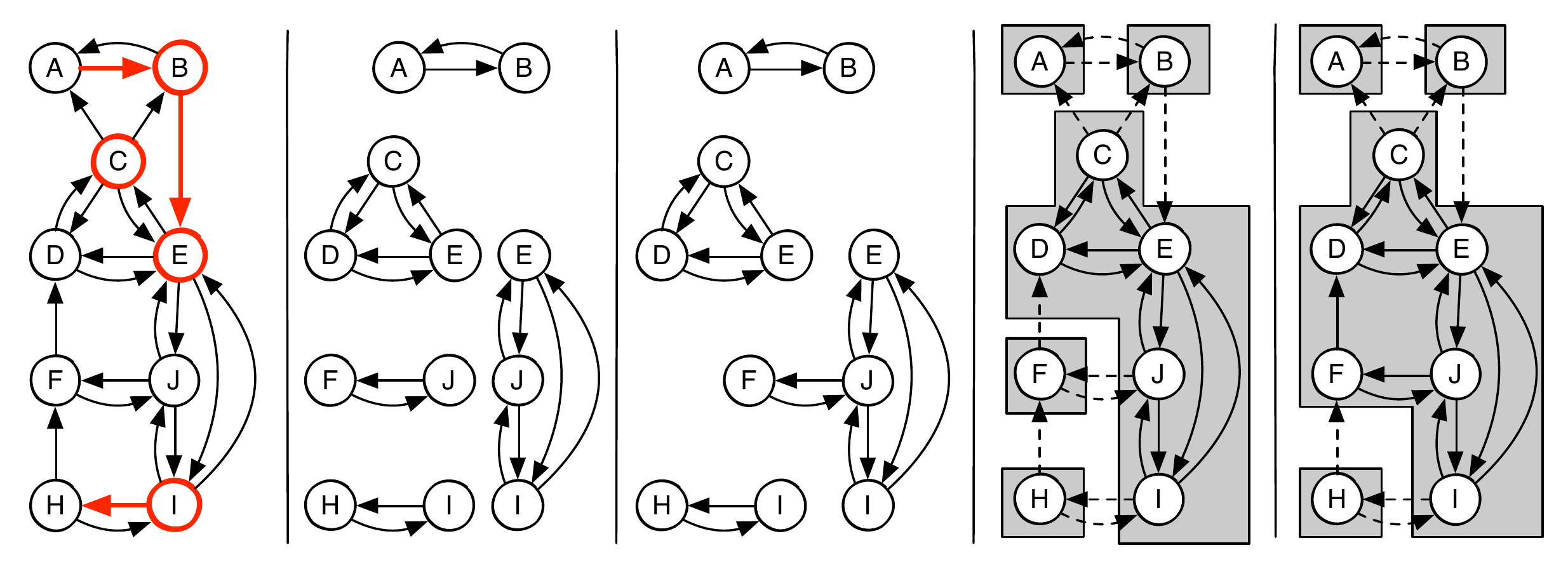}
\begin{tabular}{C{0.0\textwidth}C{0.10\textwidth}C{0.25\textwidth}C{0.15\textwidth}C{0.20\textwidth}C{0.15\textwidth}C{0.20\textwidth}}
\\
& a) $G$ &  b) $2VCC(G)$  & c) $2VCB(G)$ &  d) $2ECC(G)$ &  e) $2ECB(G)$ &
\end{tabular}
\end{center}
\vspace{-0.6cm}
\caption{(a) A strongly connected digraph $G$, with strong articulation points and strong bridges shown in red (better viewed in color). (b) The $2$-vertex-connected components of $G$ (note the three  degenerate $2$-vertex-connected components: $\{A,B\}$, $\{F,J\}$ and $\{H,I\}$). (c) The $2$-vertex-connected blocks of $G$ (note the two degenerate $2$-vertex-connected blocks: $\{A,B\}$ and $\{H,I\}$). (d) The $2$-edge-connected components of $G$.  (e) The $2$-edge-connected blocks of $G$.
}
\label{fig:graph1}
\end{figure}

\begin{figure}[t!]
\begin{center}
\vspace{-.4cm}
\includegraphics[width=\textwidth]{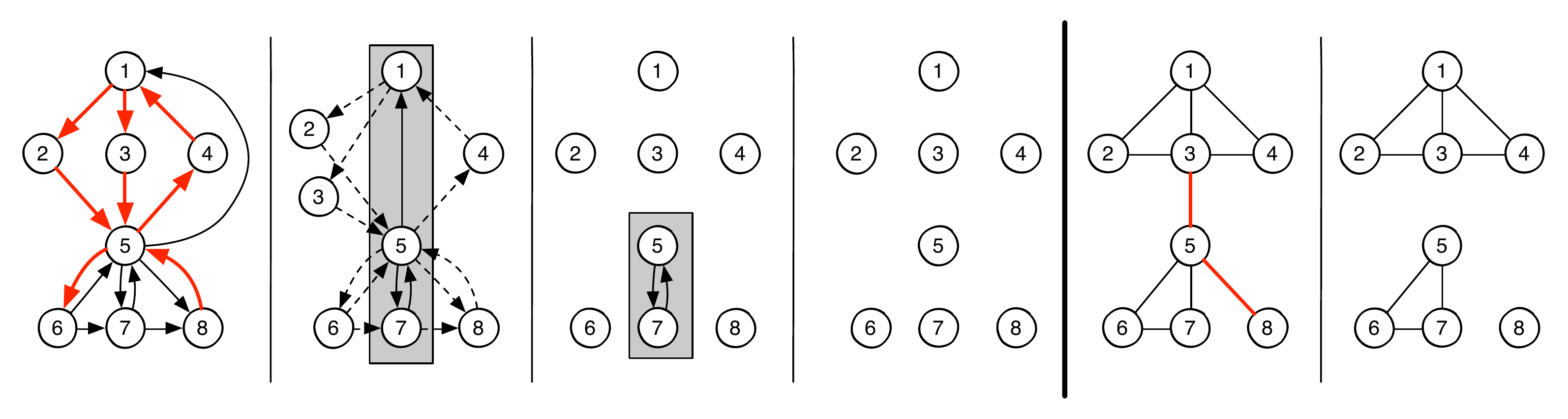}
\begin{tabular}{C{0.14\textwidth}C{0.14\textwidth}C{0.14\textwidth}C{0.14\textwidth}C{0.14\textwidth}C{0.16\textwidth}}
a) $G$ &  b) $2ECB(G)$ & c) $G/SB(G)$ & d) $2ECC(G)$ & e) $U$ & f) $2ECC(U)$
\end{tabular}
\end{center}
\vspace{-0.5cm}
\caption{(a) A digraph $G$ with strong bridges shown in red; (b) The $2$-\emph{edge-connected blocks} of $G$; (c)
The strongly connected components left after removing all the strong bridges from $G$; (d) The $2$-\emph{edge-connected components} of $G$. (e) An undirected graph $U$ with bridges shown in red; (f) The
$2$-\emph{edge-connected components} of $U$, corresponding to the
$2$-\emph{edge-connected blocks} and to
the connected components left after the removal of all bridges of $U$.
\label{fig:example1}}
\vspace{-0.5cm}
\end{figure}

It is thus not surprising that, despite being complete analogs of the corresponding notions  on undirected graphs, $2$-edge- and $2$-vertex connectivity problems appear  to be more difficult on digraphs. In particular, although  all the strong bridges and strong articulation points of a digraph can be found in linear time~\cite{Italiano2012}, computing efficiently, say in linear time, the $2$-edge- and $2$-vertex-connected components, or the $2$-edge- and $2$-vertex-connected blocks in a digraph has been an elusive goal.
A simple algorithm for computing the $2$-edge-connected components can be obtained by repeatedly removing all the strong bridges in the graph (and repeating this process until no strong bridges are left). Since at each round all the strong bridges can be computed in $O(m+n)$ time~\cite{Italiano2012} and there can be at most $O(n)$ rounds, the total time taken by this algorithm is $O(mn)$. As for $2$-vertex connectivity, Erusalimskii and Svetlov~\cite{biblocks:ES1980} proposed an algorithm that reduces the problem of computing the $2$-vertex-connected components of a digraph to the computation of the $2$-vertex-connected components in an undirected graph, but did not analyze the running time of their algorithm.
Jaberi~\cite{2vcc:jaberi14} showed that the algorithm of Erusalimskii and Svetlov has $O(nm^2)$ running time, and proposed two different algorithms with running time $O(mn)$. Both algorithms follow substantially the same high-level approach as the simple algorithm for computing the $2$-edge-connected components of a digraph.
A simple algorithm for computing the $2$-edge- or  $2$-vertex-connected blocks of a digraph takes $O(mn)$ time: given a vertex $v$, one can find in linear time all the vertices that are $2$-edge- or $2$-vertex-connected with $v$ with the help of dominator trees. Since in the worst case this step must be repeated for all vertices $v$, the total time required by the algorithm is $O(mn)$.

From the above discussion it is clear that, differently from the case of undirected graphs, for digraphs there is a huge gap between the $O(m+n)$ time bound for computing all connectivity cuts (strong bridges and strong articulation points), and the $O(mn)$ time bound for computing the connectivity blocks or components ($2$-edge- and $2$-vertex-connected blocks and $2$-edge- and $2$-vertex-connected components).
Thus, it seems quite natural to ask  whether the $O(mn)$ bound is a natural barrier for those problems, or whether they could be solved faster in linear time.

In this paper, we answer this question by presenting the first linear-time algorithm to compute the $2$-edge-connected blocks of a digraph. Our approach hinges on two different algorithms. The first is a simple iterative algorithm that builds the $2$-edge-connected blocks by removing one strong bridge at a time. The second algorithm is  more involved and recursive: the main idea is to consider simultaneously how different strong bridges partition vertices with the help of dominator trees. Although both algorithms run in $O(mn)$ time in the worst case, we show that a careful combination of the iterative and the recursive method is able to achieve the claimed linear-time bound.
Using our algorithm for $2$-edge-connected blocks, we can preprocess in linear time a digraph, and then answer in constant time queries on whether any two vertices are  $2$-edge-connected. We also show how to compute in linear time a sparse certificate for $2$-edge-connected blocks, i.e., a subgraph of the input graph that has $O(n)$ edges and maintains the same $2$-edge-connected blocks as the input graph.
Our techniques can be extended to the computation of the $2$-vertex-connected blocks of a digraph. However, in this case the low-level details become much more complicated.

\section{Flow graphs, dominators, and bridges}
\label{sec:dominators}

In this section we introduce some terminology that will be useful throughout the paper.
A \emph{flow graph} is a digraph such that every vertex is reachable from a distinguished start vertex. Let $G=(V,E)$ be the input digraph, which we assume to be strongly connected. (If not, we simply treat each strongly connected component separately.) For any vertex $s \in V$, we denote by $G(s)=(V,E,s)$ the corresponding flow graph with start vertex $s$; all vertices in $V$ are reachable from $s$ since $G$ is strongly connected. The \emph{dominator relation} in $G(s)$
is defined as follows: A vertex $u$ is a \emph{dominator} of a vertex $w$ ($u$ \emph{dominates} $w$) if every path from $s$ to $w$ contains $u$; $u$ is a \emph{proper dominator} of $w$ if $u$ dominates $w$ and $u \not= w$.  The dominator relation is reflexive and transitive. Its transitive reduction is a rooted tree, the \emph{dominator tree} $D(s)$: $u$ dominates $w$ if and only if $u$ is an ancestor of $w$ in $D(s)$. If $w \not= s$, $d(w)$, the parent of $w$ in $D(s)$, is the \emph{immediate dominator} of $w$: it is the unique proper dominator of $w$ that is dominated by all proper dominators of $w$.
An edge $(u,w)$ is a \emph{bridge} in $G(s)$ if all paths from $s$ to $w$ include $(u,w)$. \footnote{Throughout the paper, to avoid danger of ambiguity we use consistently the term  \emph{bridge} to refer to a bridge of a flow graph $G(s)$ and the term \emph{strong bridge} to refer to a strong bridge in the original graph $G$.}

Lengauer and Tarjan~\cite{domin:lt} presented an algorithm for computing dominators in  $O(m \alpha(n, m/n))$ time for a flow graph with $n$ vertices and $m$ edges, where $\alpha$ is a functional inverse of Ackermann's function~\cite{dsu:tarjan}.
Subsequently, several linear-time algorithms
were discovered~\cite{domin:ahlt,dominators:bgkrtw,domin:bkrw,dominators:Fraczak2013,dominators:poset,dom:gt04}.
Tarjan~\cite{T74TR} showed that
the bridges of flow graph $G(s)$ can be computed in $O(m)$ time given $D(s)$. He also also presented an $O(m \alpha(n, m/n))$-time algorithm to compute bridges that uses static tree set union to contract strongly connected subraphs in $G$~\cite{st:t}.
The Gabow-Tarjan static tree disjoint set union algorithm \cite{dsu:gt} reduces the running
time of this algorithm to $O(m)$ on a RAM. Buchsbaum et al.~\cite{dominators:bgkrtw} gave an $O(m)$-time pointer-machine algorithm.

Italiano et al. \cite{Italiano2012} showed that the strong articulation points of $G$ can be computed from the dominator trees of $G(s)$ and $G^R(s)$, where $s$ is an arbitrary start vertex and $G^R$ is the the digraph that results from $G$ after reversing edge directions; similarly, the strong bridges of $G$ correspond to the bridges of $G(s)$ and $G^R(s)$. This gives the following bound on the number of strong bridges.

\begin{lemma} \emph{(\cite{Italiano2012})}
\label{lemma:strong-bridges-number}
Any digraph with $n$ vertices has at most $2n-2$ strong bridges.
\end{lemma}

Experimental studies for algorithms that compute dominators, strong bridges, and strong articulation points are presented in \cite{strong-articulation:sea,Loops:SEA14}. The experimental results show that the corresponding fast algorithms given in \cite{dominators:Fraczak2013,Italiano2012,domin:lt,st:t} perform very well in practice even on very large graphs. 
\section{Computing the $2$-edge-connected blocks}
\label{section:2-edge-connected-blocks}

We recall that $u \leftrightarrow_{\mathrm{2e}} w$  denotes that vertices $u$ and $w$ are $2$-edge-connected, and that a \emph{$2$-edge-connected block} of a digraph $G=(V,E)$ is a maximal subset $B \subseteq V$ such that $u \leftrightarrow_{\mathrm{2e}} w$ for all $u, w \in B$.

\begin{theorem}
\label{theorem:2ECB}
The $2$-edge-connected blocks of a digraph $G=(V,E)$ form a partition of $V$.
\end{theorem}
\begin{proof}
We show that $\leftrightarrow_{\mathrm{2e}}$ is an equivalence relation. The relation is by definition reflexive and symmetric, so it remains to show that it is also transitive when $G$ has at least three vertices. Let $u$, $v$, and $w$ be three distinct vertices such that $u \leftrightarrow_{\mathrm{2e}} v$ and $v \leftrightarrow_{\mathrm{2e}} w$. Consider any $u$-$w$ cut $(U,W)$, where $u \in U$ and $w \in W$. Let $k'$ be the number of edges directed from $U$ to $W$.
We will show that $k' \ge 2$.  If $v \in U$, then $v \leftrightarrow_{\mathrm{2e}} w$ implies that $k' \ge 2$. Otherwise, $v \in W$, and $u \leftrightarrow_{\mathrm{2e}} v$ implies that $k' \ge 2$. The fact that $u \leftrightarrow_{\mathrm{2e}} w$ follows from Menger's Theorem \cite{menger}.
\end{proof}

Throughout, we use the notation $[v]_\mathrm{2e}$ to denote the $2$-edge-connected block containing vertex $v\in V$.
We can generalize the $2$-edge-connected relation for $k \ge 2$ edge-disjoint paths: the proof of Theorem \ref{theorem:2ECB} can be extended to show that this relation also defines a partition of $V$ into $k$-edge-connected blocks.
By Theorem \ref{theorem:2ECB}, once the $2$-edge-connected blocks are available, it is easy to test in constant time if
two vertices are $2$-edge-connected.

Next we develop algorithms that compute the $2$-edge-connected blocks of a digraph $G$. Clearly, we can assume that $G$ is strongly connected, so $m \ge n$. If not, then we process each strongly connected component separately; if $u \leftrightarrow_{\mathrm{2e}} v$ then $u$ and $v$ are in the same strongly connected component $S$ of $G$, and moreover, any vertex on a path from $u$ to $v$ or from $v$ to $u$ also belongs in $S$.
We begin with a simple algorithm that removes a single strong bridge at a time. In order to get a more efficient solution, we need to consider simultaneously how different strong bridges partition the vertex set. We present a recursive algorithm that does this with the help of dominator trees. Although both these algorithms run in $O(mn)$ time in the worst case, we finally show that a careful combination of them achieves linear time.

\subsection{A simple algorithm}
\label{section:2ECB-simple}

Algorithm \textsf{Simple2ECB} is an immediate application of the characterization of the $2$-edge-connected blocks in terms of strong bridges.
Let $u$ and $v$ be two distinct vertices. We say that a strong bridge $e$ \emph{separates $u$ from $v$} if all paths from $u$ to $v$ contain $e$. In this case $u$ and $v$ belong to different strongly connected components of $G\setminus e$.
This observation implies that we can compute the $2$-edge-connected blocks by computing the strongly connected components of $G\setminus e$ for every strong bridge $e$.

\begin{center}
\fbox{
\begin{minipage}[h]{\textwidth}
\begin{center}
\textbf{Algorithm \textsf{Simple2ECB}: Computation of the $2$-edge-connected blocks of a strongly connected digraph $G=(V,E)$}
\end{center}
\begin{description}\setlength{\leftmargin}{10pt}
\item[Step 1:] Compute the strong bridges of $G$.
\item[Step 2:] Initialize the current $2$-edge-connected blocks as $[v]_{\mathrm{2e}} = V$. (Start from the trivial partition containing only one block.)
\item[Step 3:] For each strong bridge $e$ do:
  \begin{description}\setlength{\leftmargin}{10pt}
    \item[Step 3.1:] Compute the strongly connected components $S_1,\ldots,S_k$ of $G\setminus e$.
    \item[Step 3.2:] Let $\{[v_1]_{\mathrm{2e}},\ldots,[v_l]_{\mathrm{2e}}\}$ be the current $2$-edge-connected blocks. Refine the partition into blocks by computing the intersections $[v_i]_{\mathrm{2e}} \cap S_j$ for all $i=1,\ldots,l$ and $j=1,\ldots,k$.
  \end{description}
\end{description}
\end{minipage}
}
\end{center}

\begin{lemma}
\label{lemma:2-edge-connected-partition-strong-bridges}
Algorithm \textsf{Simple2ECB} runs in $O(m b^{\ast})$ time, where $b^{\ast}$ is the number of strong bridges of $G$.
\end{lemma}
\begin{proof}
The strong bridges of $G$ can be computed in linear time by \cite{Italiano2012}. In each iteration of Step 3, we can compute the strongly connected components of $G\setminus e$ in linear-time \cite{dfs:t}.
As we discover the $i$th strongly connected component, we assign label $i$ ($i \in \{1,\ldots,n\}$) to the vertices in $S_i$.
Then, the refinement of the current blocks in Step 3.1 can be done in $O(n)$ time with bucket sorting. So each iteration takes $O(m)$ time.
\end{proof}

Note that the above bound is $O(mn)$ in the worst case, since for any digraph $b^{\ast} \le 2n-2$ by Lemma \ref{lemma:strong-bridges-number}. Despite the fact that removing a single strong bridge at a time does not yield an efficient algorithm, we will make use of this idea, in a more restricted way, in our linear-time algorithm.

\subsection{A recursive algorithm}
\label{section:2ECB-recursive}

In order to obtain a faster algorithm we need to determine how multiple strong bridges affect the partition of the vertices into blocks. We achieve this by using the dominator tree $D(s)$ of the flow graph $G(s)$, for an arbitrarily chosen start vertex $s$.
We do this as follows.
First we consider the computation of the $2$-edge-connected block that contains a specific vertex $v$. Let $w$ be a vertex other than $v$.
We say that $w$ is $2$-\emph{edge-connected from $v$} if there are two edge disjoint paths from $v$ to $w$.
Analogously, $w$ is $2$-\emph{edge-connected to $v$} if there are two edge disjoint paths from $w$ to $v$.
We divide the computation of $[v]_\mathrm{2e}$ in two parts, where the first part finds the set of vertices $[v]_{\overrightarrow{\mathrm{2e}}}$ that are 2-edge-connected from $v$, and the second part finds the set $[v]_{\overleftarrow{\mathrm{2e}}}$ of vertices that are $2$-edge-connected to $v$. Then $[v]_\mathrm{2e}$ is formed by the intersection of these two sets.

Consider the computation of $[v]_{\overrightarrow{\mathrm{2e}}}$. An efficient way to compute this set is based on dominators and bridges.
We compute the dominator tree $D(v)$ and identify the bridges of $G(v)$.
Then for each bridge $e=(u,w)$ we have $d(w)=u$, i.e., each bridge is also an edge in the dominator tree; we mark $w$ in $D(v)$.

\begin{lemma}
\label{lemma:2-connected-from}
$w \in [v]_{\overrightarrow{\mathrm{2e}}}$ if and only if $w$ is not dominated in $G(v)$ by a marked vertex.
\end{lemma}
\begin{proof}
We have that $w \not\in [v]_{\overrightarrow{\mathrm{2e}}}$ if and only if there an edge (strong bridge) that separates $v$ from $w$ in $G$. Then $e=(x,y)$ is such an edge if and only if it is a bridge in $G(v)$, so $y$ is a marked ancestor of $w$ in $D(v)$.
\end{proof}

Lemma \ref{lemma:2-connected-from} implies a straightforward linear-time algorithm to compute $[v]_{\overrightarrow{\mathrm{2e}}}$, given the dominator tree $D(v)$ of $G(v)$.
We use the same algorithm to compute $[v]_{\overleftarrow{\mathrm{2e}}}$, but operate on the reverse graph $G^R(v)$ and its dominator tree $D^R(v)$.
That is, we identify the bridges of the flow graph $G^{R}(v)$, and for each bridge $e=(u,w)$ we mark $w$ in $D^R(v)$.
Note that a vertex $w$ that is marked in $D(v)$ may not be marked in $D^R(v)$ and vice versa.

\begin{corollary}
\label{corollary:2-connected-v}
$w \in [v]_\mathrm{2e}$ if and only if $w$ is not dominated in $G(v)$ and in $G^R(v)$ by any marked vertex. Moreover, $[v]_\mathrm{2e}$ can be computed in $O(m)$ time.
\end{corollary}

Now our goal is to extend this method so that we discover all blocks $[v]_\mathrm{2e}$, without applying
Corollary \ref{corollary:2-connected-v} for all vertices $v$. Let $s$ be an arbitrarily chosen start vertex.
We first observe that the dominator trees
$D(s)$ and $D^R(s)$ of $G(s)$ and $G^R(s)$, respectively, partition the vertices into sets that contain the $2$-edge-connected blocks, as follows.
Identify the bridges of $G(s)$ (resp., $G^{R}(s)$), and for each bridge $e=(u,w)$ mark $w$ in $D(s)$ (resp. $D^R(s)$) as above.
Remove from $D(s)$ all edges $(d(v),v)$ such that $v$ is marked in $D(s)$, and remove from $D^R(s)$ all edges $(d^R(v),v)$ such that $v$ is marked in $D^R(s)$.
This decomposes the dominator trees $D(s)$ and $D^R(s)$ into forests of rooted trees, where each tree is rooted either at the start vertex $s$ or at a marked vertex. 
In the following, we use the notation $T(v)$ to denote the tree containing vertex $v$ in the decomposition of dominator tree $D(s)$. Note that $T(v)$ is a subtree of $D(s)$ and its root $r_v$ is either $s$ or a marked vertex. Similarly, we denote
by $T^R(v)$ the tree 
containing vertex $v$ in the decomposition of $D^R(s)$. In Figure~\ref{fig:subtrees} we can see an example of a flow graph $G(s)$, its dominator tree $D(s)$ and the decomposion of $D(s)$ into subtrees induced by the removal of all bridges of $G(s)$. 
The following lemma provides a necessary condition for two vertices to be $2$-edge-connected.

\begin{figure}[t!]
\begin{center}
\includegraphics[width=\textwidth]{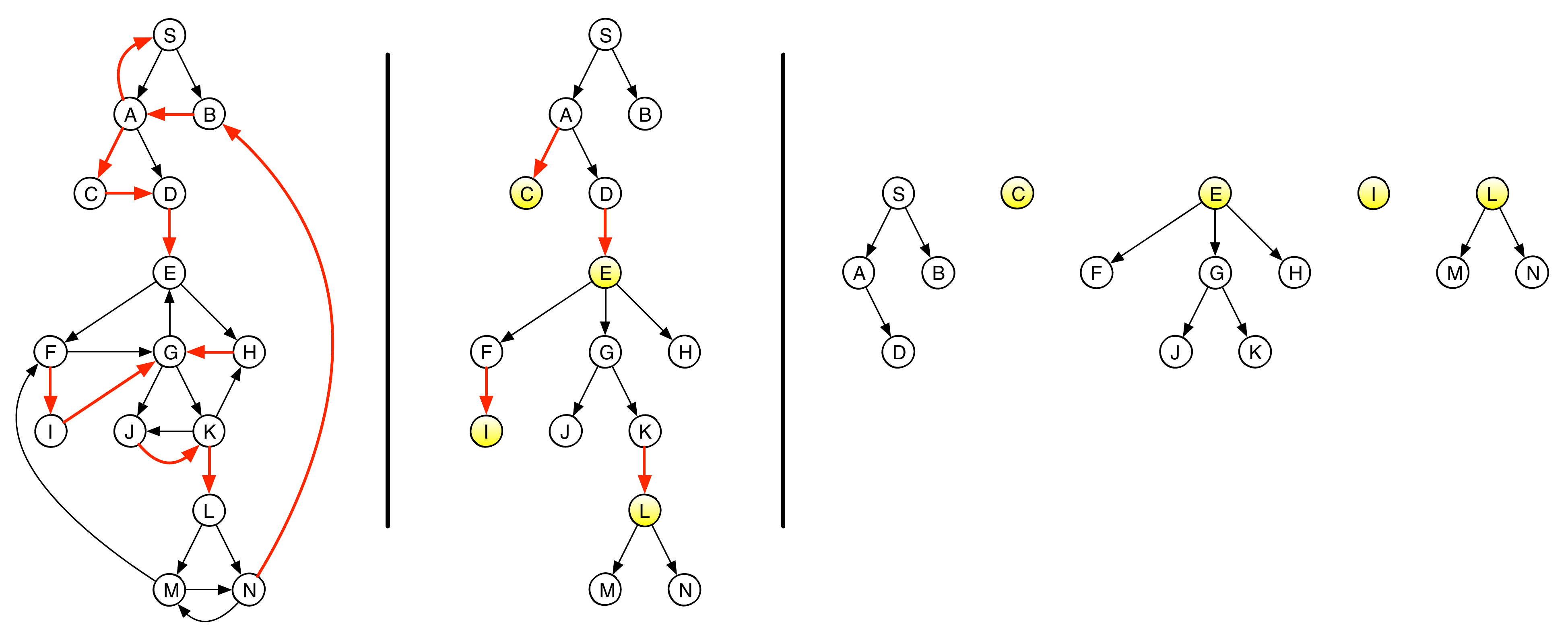}
\caption{A flow graph $G(s)$, its dominator tree $D(s)$ and its partition into the subtrees $T(v)$ induced by the bridges of $G(s)$. Strong bridges of the original graph $G$ and bridges of the flow graph $G(s)$ and are shown in red;  marked vertices are shown in yellow. (Better viewed in color.) \label{fig:subtrees}
}
\end{center}
\end{figure}

\begin{lemma}
\label{lemma:2-connectivity-necessary-condition}
$[v]_\mathrm{2e} = [w]_\mathrm{2e}$ only if $T(v)=T(w)$ and $T^R(v)=T^R(w)$.
\end{lemma}
\begin{proof}
We show that $[v]_\mathrm{2e} = [w]_\mathrm{2e}$ implies $T(v)=T(w)$. Then the same argument applied on $G^R(s)$ shows that $T^R(v)=T^R(w)$.
Suppose by contradiction that $[v]_\mathrm{2e} = [w]_\mathrm{2e}$ but $T(v)\neq T(w)$, i.e., $w \not\in T(v)$. Assume that $r_v$ is not an ancestor of $r_w$ in $D(s)$. (If $r_v$ is an ancestor of $r_w$, swap $v$ and $w$.) Let $d$ be the nearest common ancestor of $r_v$ and $r_w$ in $D(s)$. Then $d \not= r_v$ by the above assumption. Let $x$ be the shallowest marked ancestor of $r_v$ that is a descendant of $d$
in $D(s)$.
Since $[v]_\mathrm{2e} = [w]_\mathrm{2e}$, then there must be a path $P$ in $G$ from $w$ to $v$ that avoids edge $e=(d(x),x)$. Since $r_v$ is not an ancestor of $r_w$ in $D(s)$, there is a path $Q$ in $G$ from $s$ to $w$ that avoids $e$. If $v \in Q$ then the part of $Q$ from $s$ to $v$ avoids $e$, which contradicts the fact that $e$ is a bridge, i.e., it induces a cut that separates $s$ from $v$ in $G$. Otherwise, $v \not\in Q$ then $Q$ followed by $P$ ($Q\cdot P$) is a path from $s$ to $v$ that avoids $e$, a contradiction.
\end{proof}

\begin{figure}[t!]
\begin{center}
\includegraphics[width=0.6\textwidth]{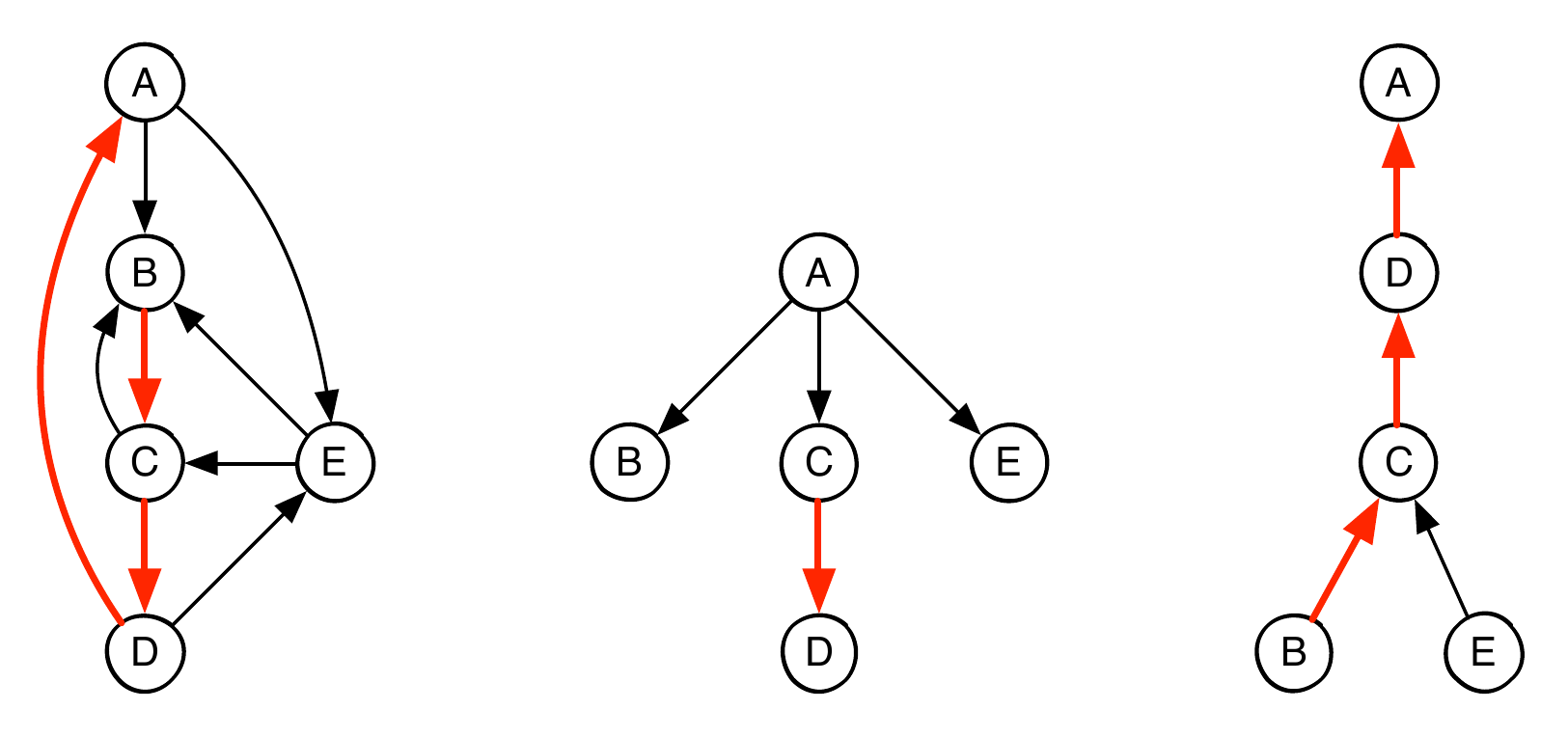}
\caption{A strongly connected digraph $G$ and its dominator trees $D(A)$ and $D^R(A)$ rooted at vertex $A$. (The edges of the dominator tree $D^R(A)$ are shown directed from child to parent.) Strong bridges are shown in red (better viewed in color). Note that vertices $C$ and $E$ lie in the same subtree in both $D(A)$ and $D^R(A)$ but they are not $2$-edge-connected, as they are separated by the strong bridge $(C,D)$.
}
\label{fig:separate}
\end{center}
\end{figure}

Note that the condition given in Lemma \ref{lemma:2-connectivity-necessary-condition} is not sufficient,  as
two vertices may be separated by a strong bridge and still lie in the same subtree in both $D(s)$ and $D^R(s)$ (see Figure \ref{fig:separate}).
The main challenge in this approach is thus to discover which vertices in the same subtree are separated by a strong bridge. To tackle this challenge, we provide some key results regarding edges and paths that connect different subtrees $T(r)$.
We will use the \emph{parent property} of dominator trees \cite{domcert}, that we state next.

\begin{lemma}
\label{lemma:parent-property} \emph{(Parent property of the dominator tree \cite{domcert}.)}
For all $(v, w) \in E$, $d(w)$ is an ancestor of $v$ in $D(s)$.
\end{lemma}

Now we prove some structural properties for paths that connect vertices in different subtrees.

\begin{lemma}
\label{lemma:parent-sibling-partition}
Let $e=(u,v)$ be an edge of $G$ such that $T(u) \not= T(v)$ and let $r_v$ be the root of $T(v)$. Then either $u=d(v)$ and $e$ is a bridge in $G(s)$, or $u$ is a proper descendant of $r_v$ in $D(s)$.
\end{lemma}
\begin{proof}
If $e$ is a bridge in $G(s)$ then $u=d(v)$ and the lemma holds. Suppose that $e$ is not a bridge, so $u \not= d(v)$.  If $v$ is an ancestor of $u$ in $D(s)$ then the lemma holds. If not, then by Lemma \ref{lemma:parent-property}, $d(v)$ is a proper ancestor of $u$ in $D(s)$.
We show that $d(v) \in T(v)$, which implies the lemma. Assume by contradiction that $d(v) \not\in T(v)$. Then $(d(v),v)$ is a bridge and $v = r_v$. Since $v$ is not an ancestor of $u$ in $D(s)$, there is a path $P$ from $s$ to $u$ that does not contain $v$. Then $P\cdot e$ is a path from $s$ to $v$ that avoids the bridge $(d(v),v)$, a contradiction.
\end{proof}

\begin{lemma}
\label{lemma:partition-paths}
Let $r$ be a marked vertex. Let $v$ be any vertex that is not a descendant of $r$ in $D(s)$. Then there is path from $v$ to $r$ that does not contain any vertex in $T(r)\setminus r$. Moreover, all simple paths from $v$ to any vertex in $T(r)$ contain the edge $(d(r),r)$.
\end{lemma}
\begin{proof}
Since $v$ is not a descendant of $r$ in $D(s)$, $v \not\in T(r)$. Graph $G$ is strongly connected, so it contains a path from $v$ to $r$. Let $P$ be any such path. Let $e=(u,w)$ be the first edge on $P$ such that $w \in T(r)$. Then by Lemma \ref{lemma:parent-sibling-partition},  either $e=(d(r),r)$ or $u$ is a proper descendant of $r$. In the first case the lemma holds. Suppose $u$ is a proper descendant of $r$. Since $v$ is not a descendant of $r$ in $D(s)$, there is a path $Q$ from $s$ to $v$ in $G$ that does not contain $r$. Then $Q$ followed by the part of $P$ from $v$ to $w$ is a path from $s$ to $w$ that avoids $d(r)$, a contradiction.
\end{proof}

We introduce the notion of \emph{auxiliary graphs} that plays a crucial role in our algorithm. It provides a decomposition of the input digraph $G$ into smaller digraphs (not necessarily subgraphs of $G$) that maintain the original $2$-edge-connected blocks.
For each subtree $T(r)$ with root $r$, such that $r$ is not a leaf in $D(s)$, we build the \emph{auxiliary graph $G_r = (V_r, E_r)$ of $r$} as follows. The vertex set $V_r$ contains a set $V_r^o$ of \emph{ordinary} vertices that are the vertices of $T(r)$, and a set $V_r^a$ of \emph{auxiliary} vertices.
The edge set $E_r$ contains all edges in $G=(V,E)$ induced by the ordinary vertices (i.e., edges $(u,v)\in E$ such that $u \in T(r)$ and $v \in T(r)$), together with some edges that have at most one endpoint in $T(r)$ and are either bridges of $G(s)$ or \emph{shortcut} edges that correspond to paths in $G$. We define those edges as follows.
Let $v$ be a vertex in $T(r)$.
We say that $v$ is a \emph{boundary vertex in} $T(r)$ if $v$ has a marked child in $D(s)$. For each marked child $w$ of $v$ in $D(s)$ we add a copy of $w$ in $V^a_r$, and add the edge $(v,w)$ in $E_r$. Also, if $r$ is marked ($r \not= s$) then we add a copy of $d(r)$ in $V_r^a$, and add the edge $(d(r),r)$ in $E_r$. We also add in $E_r$ the following shortcut edges for edges $(u,v)$ of the following type:
(a) If $u$ is ordinary and $v$ is not a descendant of $r$, then we add the shortcut edge $(u,d(r))$. (b) If $v$ is ordinary and $u$ is a proper descendant in $D(s)$ of a boundary vertex $w$, then we add the shortcut edge $(z,v)$ where $z$ is the child of $w$ that is an ancestor of $u$ in $D(s)$. (c) Finally, if $u$ is a proper descendant in $D(s)$ of a boundary vertex $w$ and $v$ is not a descendant of $r$, then we add the shortcut edge $(z,d(r))$, where $z$ is the child of $w$ that is an ancestor of $u$ in $D(s)$. We note that we do not keep multiple (parallel) shortcut edges (see Figure~\ref{fig:auxiliary}).

\begin{figure}[t!]
\begin{center}
\includegraphics[width=\textwidth]{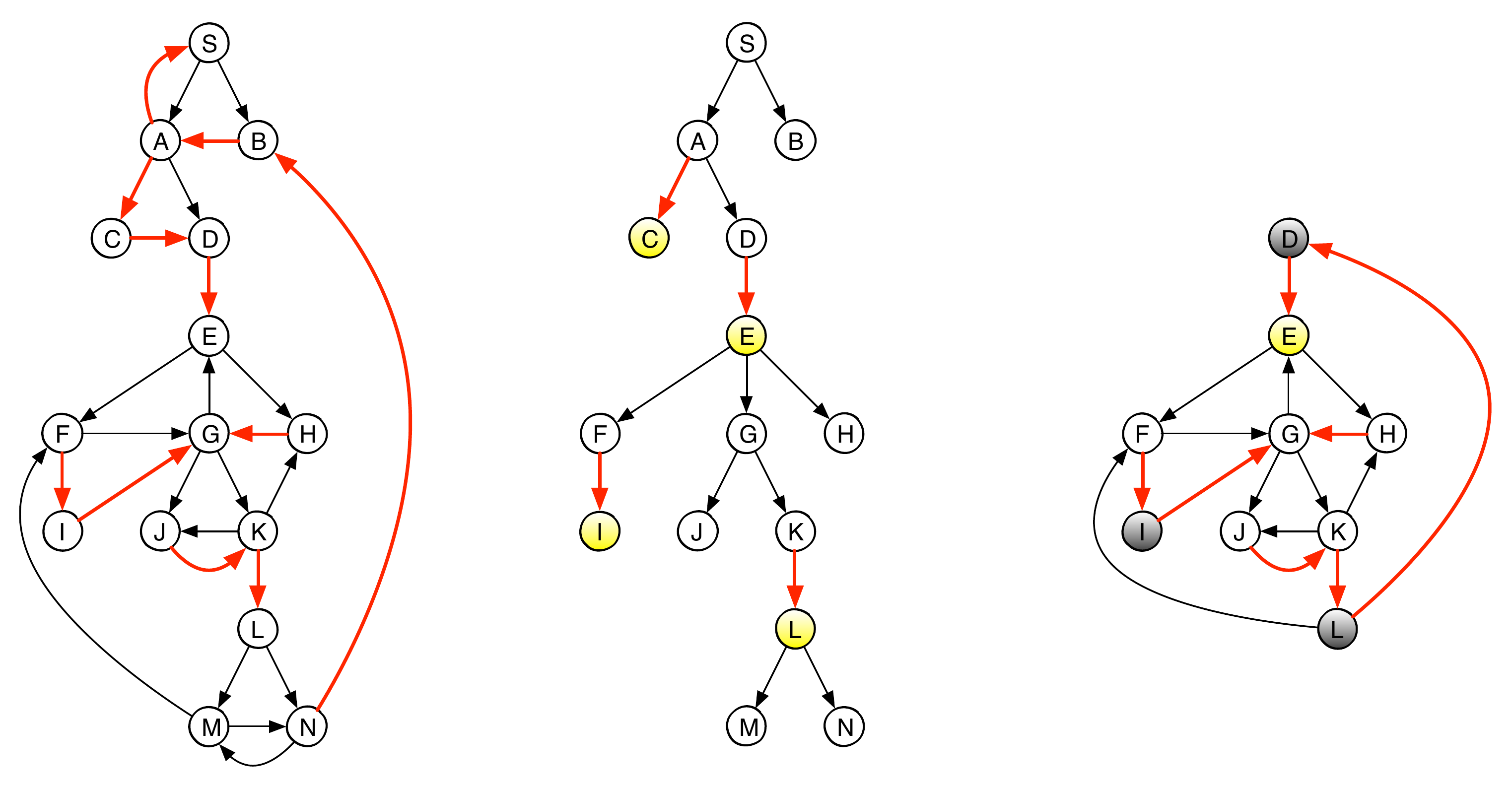}
\caption{The flow graph $G(S)$ and its dominator tree $D(S)$ from Figure~\ref{fig:subtrees}, together with the auxiliary graph of vertex $E$. Strong bridges are red, marked vertices are yellow, and auxiliary vertices are gray.  (Better viewed in color.)  Edge $(L,D)$ is a shortcut edge that corresponds to a path in $G$ from $L$ to $D$, e.g., $L, N, B, A, D$.
}
\label{fig:auxiliary}
\end{center}
\end{figure}

\begin{lemma}
\label{lemma:partition-subgraphs-size}
If $G(s)$ has $b$ bridges then the auxiliary graphs $G_r$ have at most $n+2b$ vertices and $m+2b$ edges in total.
\end{lemma}
\begin{proof}
Every vertex appears as a ordinary vertex in one auxiliary graph. A marked vertex in $D(s)$ corresponds to a bridge in $G(s)$, so there are $b \le n-1$ marked vertices.
Since we have one auxiliary graph for each marked vertex, the total number of the auxiliary vertices $d(r)$ is $b$.
Each marked vertex $v$ can also appear in at most one other auxiliary graph as a child of a boundary vertex. So the total number of vertices is at most $n+2b$.
Next we bound the total number of edges.
The total number of edges between ordinary vertices in each $G_v$ is at most $m-b$.
Each bridge can appear in at most two auxiliary graphs. Finally, the number of edges connecting auxiliary vertices is at most $b$, since each such edge corresponds to a unique copy of a marked vertex. So we have at most $m+2b$ edges in total.
\end{proof}

Next we show that we can compute the $2$-edge-connected blocks in each auxiliary graph independently of each other.

\begin{lemma}
\label{lemma:partition-auxiliary-paths}
Let $v$ and $w$ be two vertices in a subtree $T(r)$. Any path from $v$ to $w$ in $G$ has a corresponding path from $v$ to $w$ in $G_r$, and vice versa.
\end{lemma}
\begin{proof}
Consider a path $P$ from $v$ to $w$ in $G$. We show that it has a corresponding path $P_r$ from $v$ to $w$ in $G_r$.
If $P$ consists only of vertices in $T(r)$ then we have $P_r = P$. Otherwise,
let $(u,x)$ be the first edge on $P$ such that $u \in T(r)$ and $x \not\in T(r)$. Also let $(y,z)$ be the first edge on $P$ after $(u,x)$ such that $y \not\in T(r)$ and $z \in T(r)$.  By Lemma \ref{lemma:parent-sibling-partition}, edge $(u,x)$ is either a bridge or $u$ is a proper descendant of $r_x$ (the root of $T(x)$ in $D(s)$). Similarly, edge $(y,z)$ is either a bridge or $y$ is a proper descendant of $r$ in $D(s)$.
Suppose $(u,x)$ is a bridge. Then $u=d(x)$. Let $t$ be the first vertex on $P$ after $x$ that is not a descendant of $x$. If $t \in T(r)$ then $t=z$. In this case the part of $P$ from $x$ to $z$ corresponds to the edge $(x,z)$ in $P_r$. If $t \not\in T(r)$ then Lemma \ref{lemma:parent-sibling-partition} implies that $t$ is not a descendant of $r$ in $D(s)$.
By Lemma \ref{lemma:partition-paths}, we have that $(y,z)=(d(r),r)$, so the the part of $P$ from $x$ to $z$ corresponds to the edge $(x,d(r))$ in $P_r$. Now suppose that $u$ is a proper descendant of $r_x$ in $D(s)$. Then Lemma \ref{lemma:partition-paths} implies that $(y,z)=(d(r),r)$, so the the part of $P$ from $u$ to $y=d(r)$ corresponds to the edge $(x,d(r))$ in $P_r$.
We can repeat the same argument for every part of $P$ that is outside $T(r)$, which gives a valid path $P_r$ in $G_r$.

Now we prove that any path $P_r$ from $v$ to $w$ in $G_r$ has a corresponding path from $v$ to $w$ in $G$. If $P_r$ consists only of vertices in $T(r)$ then we have $P = P_r$. Otherwise,
let $(u,x)$ be the first edge on $P_r$ such that $u \in T(r)$ and $x \not\in T(r)$. Also let $(y,z)$ be the first edge on $P_r$ after $(u,x)$ such that $y \not\in T(r)$ and $z \in T(r)$. Then $x, y \in V_r^a$.
By construction, $(u,x)$ is either $(u,d(r))$ or $(d(x),x)$.
In the first case, $y=d(r)$ and $z=r$, since $(d(r),r)$ is the only edge leaving $d(r)$. Also $G$ contains an edge $(u,w)$ of type (a), where $w$ is not a descendant of $r$ in $D(s)$. Thus, by Lemma \ref{lemma:partition-paths} there is a path $Q$ in $G$ from $w$ to $r$ that does not contain any vertex in $T(r)\setminus r$ and contains the edge $(d(r),r)$.
This path corresponds to the part of $P_r$ that consists of the edges $(u,d(r))$ and $(d(r),r)$.
In the second case, $(u,x)=(d(x),x)$. Suppose $(y,z)=(d(r),r)$. Then $y=x$ since the edges leaving $x$ can only enter $T(r)$ or $d(r)$.
This implies that $G$ contains an edge $(q,t)$ of type (c), where $q$ is a descendant of $x$ and $t$ is not a descendant of $r$ in $D(s)$. By Lemma \ref{lemma:parent-property}, there is a path $Q$ in $G$ from $x$ to $q$ that contains only descendants of $x$ in $D(s)$.  Also, by Lemma \ref{lemma:partition-paths}, there is a path $Q'$ in $G$ from $t$ to $r$ that contains $(d(r),r)$. Path $Q\cdot (q,t)\cdot Q'$ is a path from $x$ to $r$ in $G$ that corresponds to the part of $P_r$ consisting of the edges $(x,d(r))$ and $(d(r),r)$.
Finally, suppose $z \in T(r)$ and $y \not= d(r)$. Then $y=x$, and $G$ contains an edge $(q,z)$ of type (b), where $q$ is a descendant of  $x$ in $D(s)$.
By Lemma \ref{lemma:parent-property}, there is a path $Q$ in $G$ from $x$ to $q$. Path $Q\cdot (q,z)$ is a path from $x$ to $z$ in $G$ that corresponds to the edge $(x,z)$ on $P_r$.
\end{proof}

\begin{corollary}
\label{corollary:auxialiry-graphs}
Each auxiliary graph $G_r$ is strongly connected.
\end{corollary}
\begin{proof}
Follows immediately from Lemma \ref{lemma:partition-auxiliary-paths} and the fact that $G$ is strongly connected.
\end{proof}

\begin{lemma}
\label{lemma:partition-sugraphs}
Let $v$ and $w$ be any two distinct vertices of $G$. Then $v$ and $w$ are $2$-edge-connected in $G$ if and only if they are
both ordinary vertices in an auxiliary graph $G_r$ and they are
$2$-edge-connected in $G_r$.
\end{lemma}
\begin{proof}
From Lemma \ref{lemma:2-connectivity-necessary-condition}, we have that $v$ and $w$ must belong in the same subtree $T(r)$, so they are both ordinary vertices of $G_r$.
Clearly if all paths from $v$ to $w$ in $G_r$ contain a common edge (strong bridge), then so do all paths from $v$ to $w$ in $G$ by Lemma \ref{lemma:partition-auxiliary-paths}. Now we prove the converse. Suppose all paths from  $v$ to $w$ in $G$ contain a common edge $e=(x,y)$. If $x,y \in T(r)$ then also all paths from $v$ to $w$ in $G_r$ contain $e$. Suppose $x \in T(r)$ and $y \not\in T(r)$. By Lemma \ref{lemma:parent-sibling-partition} either $x=d(y)$ or $x$ is a descendant of $r_y$. In the former case, all paths from $v$ to $w$ in $G_r$ contain $e$. In the latter, Lemma \ref{lemma:partition-paths} implies that all paths from $v$ to $w$ in $G$ contain $(d(r),r)$. By Lemma \ref{lemma:partition-auxiliary-paths} this is also true for all paths from $v$ to $w$ in $G_r$.
Next consider that $x \not\in T(r)$ and is a descendant of $r$. Then $v$ is not an ancestor of $w$ in $D(s)$, since otherwise, by Lemma \ref{lemma:parent-property}, there would be a path from $v$ to $w$ that avoids $e$. Let $w \in T(r)$ be the boundary vertex that is an ancestor of $x$, and let $z$ be the child of $w$ that is an ancestor of $x$. By Lemma \ref{lemma:partition-paths}, all paths from $v$ to $x$ in $G$, and thus all paths from $v$ to $w$, contain the bridge $(w,z)$. By Lemma \ref{lemma:partition-auxiliary-paths} this is also true for all paths from $v$ to $w$ in $G_r$.
Finally, if $x \not\in T(r)$ and is not a descendant of $r$, Lemma \ref{lemma:partition-paths} implies that all paths from $x$ to $w$ in $G$ contain the bridge $(d(r),r)$.  Hence, all paths from $v$ to $w$ in $G$ contain the bridge $(d(r),r)$, and so do all paths from $v$ to $w$ in $G_r$ by Lemma \ref{lemma:partition-auxiliary-paths}.
\end{proof}

To construct the auxiliary graphs $G_r = (V_r, E_r)$ we need to specify how to compute the shortcut edges of each type (a), (b), and (c). Suppose $(u,v)$ is an edge of type (a). Then $v$ is not a descendant of $r$ in $D(s)$, which can be tested using an $O(1)$-time test of the ancestor-descendant relation. There are several simple $O(1)$-time tests of this relation~\cite{dfs:t}. The most convenient one for us is to number the vertices of $D(s)$ from $1$ to $n$ in preorder, and to compute the number of descendants of each vertex $v$, which we denote by $\mathit{size}(v)$. Then $v$ is a descendant of $r$ if and only if $\mathit{pre}(r) < \mathit{pre}(v) < \mathit{pre}(r) + \mathit{size}(r)$. Next suppose that $(u,v)$ is of type (b). Then $u$ is a proper descendant of a boundary vertex $w$ in $D(s)$. To compute the shortcut edge of $(u,v)$ we need to find the child $z$ of $w$ that is an ancestor of $u$ in $D(s)$. To that end, we create a list $B_r$ that contains the edges $(u,v)$ of type (b) such that $v \in T(r)$, and sort $B_r$ in increasing preorder of $u$. We create a second list $B'_r$ that contains the children in $D(s)$ of the boundary vertices in $T(r)$, and sort $B_r$ in increasing preorder. Then, the shortcut edge of $(u,v)$ is $(z,v)$, where $z$ is the last vertex in the sorted list $B'_r$ such that $\mathit{pre}(z) \le \mathit{pre}(u)$. Thus the shortcut edges of type (b) can be computed in linear time by bucket sorting and merging. Finally, consider the edges of type (c). For each such edge $(u,v)$ we need to add the edge $(z,d(r))$ in each $G_r$, where $u$ is a proper descendant of a boundary vertex $w \in T(r)$, $v$ is not a descendant of $r$ in $D(s)$, and $z$ is the child of $w$ that is an ancestor of $u$ in $D(s)$. We compute these edges for all auxiliary graphs $G_r$ as follows. First, we create a compressed tree $\widehat{D}(s)$ that contains only $s$ and the marked vertices. A marked vertex $v$ becomes child of its nearest marked ancestor $u$, or of $s$ if $u$ does not exist.
This easily done in $O(n)$ time during the preorder traversal of $D(s)$.
Next we process all edges $(u,v)$ such that $v$ is not a descendant of $r_u$ in $D(s)$. At each node $w \not= s$ in $\widehat{D}(s)$ we store a label $\ell(w)$ which is the minimum $\mathit{pre}(r_v)$ of an edge $(u,v)$ of type (c) such that $u \in T(w)$; we let $\ell(w)=\mathit{pre}(w)$ if no such edge exists. Using these labels we compute for each $w \not= s$ in $\widehat{D}(s)$ the values $\mathit{low}(w) = \min\{ \ell(v) \ | \ v \mbox{ is a descendant of } w \mbox{ in } \widehat{D}(s) \}$. These computations can be done in $O(m)$ time by processing the tree $\widehat{D}(s)$ in a bottom-up order. Now consider the auxiliary graph $G_r$. We process the children in $D(s)$ of the boundary vertices in $T(r)$. Note that these children are marked, so they have a $\mathit{low}$ value. For each such child $z$ we test if $G_r$ has a shortcut edge $(z,d(r))$: If $\mathit{low}(z) < \mathit{pre}(r)$ then we add the edge $(z,d(r))$.

\begin{lemma}
\label{lemma:partition-subgraphs-construct}
We can compute all auxiliary graphs $G_r$ in $O(m)$ time.
\end{lemma}

Lemma \ref{lemma:partition-sugraphs} allows us to compute the $2$-edge-connected blocks of each auxiliary graph separately.
Algorithm \textsf{Rec2ECB} applies this idea recursively, until all ordinary vertices in each auxiliary graph are $2$-edge-connected.  Since the auxiliary vertices of $G_r$ do not belong to the same set of the $2$-edge-connected partition as the ordinary vertices of $G_r$, we only need to consider the bridges that separate ordinary vertices.

\begin{figure}[t!]
\begin{center}
\fbox{
\begin{minipage}[h]{\textwidth}
\begin{center}
\textbf{Algorithm \textsf{Rec2ECB}: Recursive computation of the $2$-edge-connected blocks for the ordinary vertices of a strongly connected digraph $G=(V,E)$}
\end{center}
\begin{description}\setlength{\leftmargin}{10pt}
\item[Step 1:] Choose an arbitrary ordinary vertex $s \in V^o$ as a start vertex. Compute the dominator trees $D(s)$ and $D^R(s)$ and the bridges of $G(s)$ and $G^R(s)$.
\item[Step 2:] Compute the number $b$ of bridges $(x,y)$ in $G(s)$ such that $y$ is an ancestor of an ordinary vertex in $D(s)$. Compute the number $b^R$ of bridges $(x,y)$ in $G^R(s)$ such that $y$ is an ancestor of an ordinary vertex in $D^R(s)$.
\item[Step 3:] If $b = b^R = 0$ then return $G$. ($[s]_\mathrm{2e}=V^o$.)
\item[Step 4:] If $b^R > b$ then swap $G$ and $G^R$. Partition $D(s)$ into subtrees $T(r)$ and compute the corresponding auxiliary graphs $G_r$. Compute recursively the $2$-edge-connected partition for each subgraph $G_r$ with at least two ordinary vertices.
\end{description}
\end{minipage}
}
\end{center}
\end{figure}

\begin{lemma}
\label{lemma:partition-algorithm}
Algorithm \textsf{Rec2ECB} runs in $O(mn)$ time.
\end{lemma}
\begin{proof}
Each recursive call refines the current partition of $V$, thus we have at most $n-1$ recursive calls. By \cite{dominators:bgkrtw,T74TR} and Lemma \ref{lemma:partition-subgraphs-construct}, the total work per recursive call is $O(m)$.
\end{proof}

\begin{figure}[t!]
\begin{center}
\includegraphics[width=\textwidth]{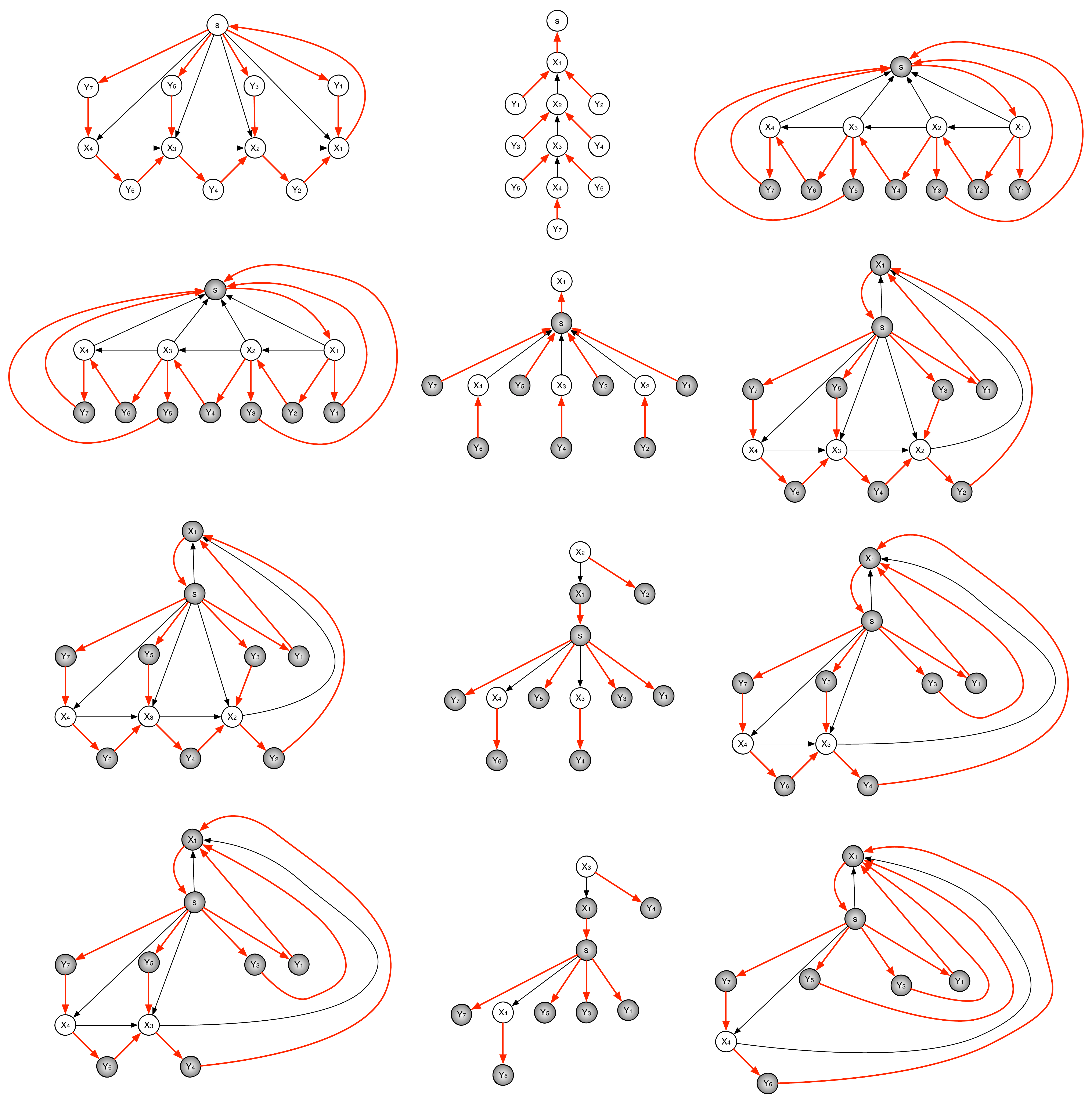}
\caption{An input digraph with $n=\Theta(k)$ vertices that causes $k$ recursive calls of Algorithm \textsf{Rec2ECB}. Left column: The input digraph in each recursive call; middle column: The dominator tree used to compute the next partition; right column: The auxiliary graph that contains the majority of ordinary vertices, that will be the input digraph in the next recursive call.
Vertices $X_1$, $X_2$, \ldots, $X_k$ are not $2$-edge-connected but Algorithm \textsf{Rec2ECB} requires $k$ recursive calls to separate them into different blocks. (In this figure $k=4$.)
}
\label{fig:badcase}
\end{center}
\end{figure}

We note that the bound stated in Lemma \ref{lemma:partition-algorithm} is tight. The same strong bridge can be used repeatedly to separate different pairs of vertices in successive recursive calls (see Figure \ref{fig:badcase}).

\subsection{Linear-time algorithm}
\label{sec:linear-time-2ECB}

Although Algorithms \textsf{Simple2ECB} and \textsf{Rec2ECB} run in $O(mn)$ time, we show that a careful combination of them gives a linear-time algorithm. The critical observation, proved in Lemma~\ref{lemma:Fast2ECB-correctness} below,  is that if a strong bridge separates different pairs of vertices in successive recursive calls (which causes the worst-case behavior of  Algorithm \textsf{Rec2ECB}), then it will appear as the strong bridge entering the root of a subtree in our decomposition of a dominator tree. Algorithm \textsf{Fast2ECB} applies this observation together with all the building blocks we developed in the previous paragraphs, and achieves the computation of the $2$-edge-connected blocks in linear time.

\begin{figure}[t!]
\begin{center}
\fbox{
\begin{minipage}[h]{\textwidth}
\begin{center}
\textbf{Algorithm \textsf{Fast2ECB}: Linear-time computation of the $2$-edge-connected blocks of a strongly connected digraph $G=(V,E)$}
\end{center}
\begin{description}\setlength{\leftmargin}{10pt}
\item[Step 1:] Choose an arbitrary vertex $s \in V$ as a start vertex. Compute the dominator tree $D(s)$ and the bridges of $G(s)$.
\item[Step 2:] Partition $D(s)$ into subtrees $T(r)$ and compute the corresponding auxiliary graphs $G_r$.
\item[Step 3:] For each auxiliary graph $H = G_r$ do:
    \begin{description}\setlength{\leftmargin}{10pt}
    \item[Step 3.1:] Compute the dominator tree $D_H^R(r)$ and the bridges of $H^R(r)$. Let $d_H^R(q)$ be the parent of $q \not= r$ in $D_H^R(r)$.
    \item[Step 3.2:] Partition $D_H^R(r)$ into subtrees $T_H^R(q)$ and compute the corresponding auxiliary graphs $H^R_q$.
    \item[Step 3.3:] For each auxiliary graph $H^R_q$ do:
        \begin{description}\setlength{\leftmargin}{10pt}
        \item[Step 3.3.1:] Compute the strongly connected components $S_1, S_2, \ldots, S_k$ of $H^R_q\setminus (d_H^R(q),q)$.
        \item[Step 3.3.2:] Partition the ordinary vertices of $H_q$ into blocks according to each $S_j$, $j=1,\ldots,k$; For each ordinary vertex $v$, $[v]_\mathrm{2e}$ contains the ordinary vertices in the strongly connected component of $v$.
        \end{description}
    \end{description}
\end{description}
\end{minipage}
}
\end{center}
\end{figure}

\begin{lemma}
\label{lemma:Fast2ECB-correctness}
Algorithm \textsf{Fast2ECB} is correct.
\end{lemma}
\begin{proof}
Let $u$ and $v$ be any vertices. If $u$ and $v$ are $2$-edge-connected in $G$, then by Lemma \ref{lemma:partition-sugraphs} they are $2$-edge-connected in both auxiliary graphs of $G$ and $G_r$ that contain them as ordinary vertices. This
 implies that the algorithm will correctly include them in the same block. So suppose that $u$ and $v$ are not $2$-edge-connected.
Then, without loss of generality, we can assume that all  paths from $u$ to $v$ contain a common strong bridge. We argue that the blocks of $u$ and $v$ will be separated in some step of the algorithm.
If $u$ and $v$ are located in different subtrees of $D(s)$ then the claim is true. If they are in the same subtree then they appear in an auxiliary graph $H=G_r$ as ordinary vertices. By Lemma \ref{lemma:partition-sugraphs}, $H$ contains a strong bridge that is contained in all paths from $u$ to $v$. Let $H^R$ be the reverse graph of $H$. Let $D_H^R(r)$ be the dominator tree of $H^R(r)$. If $u$ and $v$ are located in different subtrees of $D_H^R$ then the claim is true. Suppose then that they are located in a subtree with root $q$. By Corollary \ref{corollary:2-connected-v}, $q \not=r$. Let $p=d_H^R(q)$ be the parent of $q$ in $D_H^R(r)$. Then $(q,p)$ is a strong bridge of $H$.  We claim that $H\setminus (q,p)$ does not contain any path from $u$ to $v$. 
To prove the claim, we consider two cases. First suppose that all paths from $v$ to $u$ in $H^R$ contain a bridge $(d_H^R(x),x)$ of $D_H^R(r)$ such that $x$ is ancestor of $u$. Then $(q,p)$ must appear in all paths from $u$ to $v$ in $H$. If not, then $(p,q) \not= (d_H^R(x),x)$, and there is a path $\pi$ in $H^R$ from $x$ to $u$ that avoids $(p,q)$. Since $x$ is an ancestor of $p$, there is a path $\pi'$ in $H^R$ from $r$ to $x$ that also avoids $(p,q)$. So $\pi' \cdot \pi$ gives a path from $r$ to $u$ in $H^R$ that avoids $(p,q)$, a contradiction.
Now suppose that there is no bridge $(d_H^R(x),x)$ of $D_H^R(r)$ with $x$ an ancestor of $u$ that is contained in all paths from $v$ to $u$ in $H^R$. Let $e$ be a strong bridge that separates $u$ from $v$ in $H$. Then $e \not= (q,p)$, so there is a path $\pi$ in $H$ from $u$ to $r$ that avoids $e$. But $H$ contains a path $\pi'$ from $r$ to $v$ that avoids $e$. Then $\pi\cdot\pi'$ is a path from $u$ to $v$ in $H$ that does not contain $e$, a contradiction.
\end{proof}

Finally, we show that the algorithm indeed runs in linear time.

\begin{lemma}
\label{lemma:Fast2ECB-time}
Algorithm \textsf{Fast2ECB} runs in $O(m)$ time.
\end{lemma}
\begin{proof}
We analyze the total time spent on each step that Algorithm \textsf{Fast2ECB} executes.
Step 1 takes $O(m)$ time by \cite{dominators:bgkrtw}, and Step 2 takes $O(m)$ time by Lemma \ref{lemma:partition-subgraphs-construct}. From Lemma \ref{lemma:partition-subgraphs-size} we have that the total number of vertices and the total number of edges in all auxiliary graphs $H$ of $G$ are $O(n)$ and $O(m)$ respectively. Therefore, the total number of strong bridges in these auxiliary graphs is $O(n)$ by Lemma \ref{lemma:strong-bridges-number}. Then, by Lemma \ref{lemma:partition-subgraphs-size}, the total size (number of vertices and edges) of all auxiliary graphs $H_q^R$ for all $H$, computed in Step 3.2, is still $O(m)$ and they are also computed in $O(m)$ total time by Lemma \ref{lemma:partition-subgraphs-construct}. So Steps 3.1 and 3.3 take $O(m)$ time in total as well.
\end{proof} 
\section{Sparse certificate for the $2$-edge-connected blocks}
\label{section:sparse-certificate}

We now show how to compute in linear time  a sparse certificate for the $2$-edge-connected blocks, i.e., a subgraph $C(G)$ of the input graph $G$ that has $O(n)$ edges
and maintains the same $2$-edge-connected blocks as the input graph.  Such a sparse certificate allows allows us to speed up computations, such as finding the actual edge-disjoint paths that connect a pair of vertices. See, e.g., \cite{sparse-k-connected:ni}.
As in Section \ref{section:2-edge-connected-blocks} we can assume without loss of generality that $G$ is strongly connected, in which case subgraph $C(G)$ will also be strongly connected.
The certificate uses the concept of \emph{independent spanning trees} \cite{domcert}. In this context, a spanning tree $T$ of a flow graph $G(s)$ is a tree with root $s$ that contains a path from $s$ to $v$ for all vertices $v$.  Two spanning trees $B$ and $R$ rooted at $s$ are \emph{independent} if for all $v$, the paths from $s$ to $v$ in $B$ and $R$ share only the dominators of $v$. Every flow graph $G(s)$ has two such spanning trees, computable in linear time \cite{domcert}. Moreover, the computed spanning trees are\emph{ maximally edge-disjoint}, meaning that the only edges they have in common are the bridges of $G(s)$.

The sparse certificate can be constructed during the computation of the $2$-edge-connected blocks, by extending Algorithm \textsf{Fast2ECB}. We now sketch the main modifications needed. During the execution of Algorithm \textsf{Fast2ECB}, we maintain a list (multiset) $L$ of the edges to be added in $C(G)$. The same edge may be inserted into $L$ multiple times, but the total number of insertions will be $O(n)$. Then we can use radix sort to remove duplicate edges in $O(n)$ time.
We initialize $L$ to be the empty.
During Step 1 of Algorithm \textsf{Fast2ECB} we compute two independent spanning trees, $B(G(s))$ and $R(G(s))$ of $G(s)$ and insert their edges into $L$.
We also add the edges of a spanning tree of the reverse flow graph $G^R(s)$.
Next, in Step 3.1 we compute two independent spanning trees $B(H^R(r))$ and $R(H^R(r))$ for each auxiliary graph $H^R(r)$.
For each edge $(u,v)$ of these spanning trees, we insert a corresponding edge into $L$ as follows. If both $u$ and $v$ are ordinary vertices in $H^R(r)$, we insert $(u,v)$ into $L$ since it is an original edge of $G$.
Otherwise, $u$ or $v$ is an auxiliary vertex and we insert into $L$ a corresponding original edge of $G$. Such an original edge can be easily found during the construction of the auxiliary graphs.
Finally,  in Step 3.3, we compute two spanning trees for every connected component $S_i$ of each auxiliary graph $H^R_q\setminus (p,q)$ as follows. Let $H_{S_i}$ be the subgraph of $H_q$ that is induced by the vertices in $S_i$. We choose an arbitrary vertex $v \in S_i$ and compute a spanning tree of $H_{S_i}(v)$ and a spanning tree of $H^R_{S_i}(v)$. We insert in $L$ the original edges that correspond to the edges of these spanning trees.

\begin{lemma}
The sparse certificate $C(G)$ has the same $2$-edge-connected blocks as the input digraph $G$.
\end{lemma}
\begin{proof}
It suffices to show that the execution of Algorithm \textsf{Fast2ECB} on $C(G)$ and produces the same $2$-edge-connected blocks as the execution of Algorithm \textsf{Fast2ECB} on $G$. The correctness of Algorithm \textsf{Fast2ECB} implies that it produces the same result regardless of the choice of start vertex $s$. So we assume that both executions choose the same start vertex $s$. We will refer to the execution of Algorithm \textsf{Fast2ECB} with input $G$ (resp. $C(G)$) as \textsf{Fast2ECB}$(G)$ (resp. \textsf{Fast2ECB}$(C(G))$).

First we note that $C(G)$ is strongly connected since it contains a spanning tree of $G(s)$ and a spanning tree of $G^R(s)$. Moreover, the fact that $C(G)$ contains two independent spanning trees of $G$ implies that $G$ and $C(G)$ have the same dominator tree and bridges with respect to the start vertex $s$ that are computed in Step 1. Hence, the subtrees $T(r)$ computed of Step 2 of Algorithm \textsf{Fast2ECB} are the same in both executions \textsf{Fast2ECB}$(G)$ and \textsf{Fast2ECB}$(C(G))$.
The same argument as in Step 1 implies that in Step 3.1, both executions \textsf{Fast2ECB}$(G)$ and \textsf{Fast2ECB}$(C(G))$ compute the same partitions $T^R(r)$ of each auxiliary graph $H^R(r)$.
Finally, by construction, the strongly connected components of each auxiliary graph $H^R_q\setminus (p,q)$ are the same in both executions of \textsf{Fast2ECB}$(G)$ and \textsf{Fast2ECB}$(C(G))$.

We conclude that \textsf{Fast2ECB}$(G)$ and \textsf{Fast2ECB}$(C(G))$ compute the same $2$-edge-connected blocks as claimed.
\end{proof} 
\section{Concluding remarks and open problems}
\label{section:concluding}

We studied $2$-edge connectivity in directed graphs and, in particular, we presented a linear-time algorithm for the $2$-\emph{edge-connected} relation among vertices. Our approach is based on a careful combination of two $O(mn)$ algorithms. Given the $2$-edge-connected blocks of a digraph $G$, it is straightforward to check in constant time if any two vertices are $2$-edge-connected.
We have implemented the algorithms described in this paper and performed preliminary experiments on large graphs (with millions of vertices and edges); in those experiments two of our algorithms, \textsf{Rec2ECB} and \textsf{Fast2ECB}, performed very well. Our techniques can be extended to the computation of the $2$-vertex-connected blocks of a directed graph. We can show that, although the $2$-vertex-connected blocks do not define a partition of the vertices, they can be represented by a tree structure with $O(n)$ nodes similar to a representation used in \cite{onlineBiconnected:WT92} for the biconnected components of an undirected graph. Using this representation we can test in constant time if any two vertices are $2$-vertex-connected.
We leave as an open question if the $2$-edge-connected or the $2$-vertex-connected components of a digraph can be computed in linear time. The best current bound for both problems is $O(mn)$.


\end{document}